\documentclass[10pt,twocolumn,twoside]{IEEEtran}
\usepackage{amsmath}
\usepackage{color}
\usepackage{cases}
\usepackage{amsfonts,amsmath,amssymb}
\usepackage{mathrsfs}
\usepackage{cite}
\usepackage{float}
\usepackage{graphicx}
\usepackage{algorithm}
\usepackage{algpseudocode}
\usepackage{url}
\usepackage{enumerate}
\usepackage{subfigure}
\usepackage{hyperref}
\usepackage{cite}

\newtheorem{theorem}{Theorem}
\newtheorem{proposition}{Proposition}

\newtheorem{example}{Example}

\newtheorem{definition}{Definition}

\begin{document}

\title{Minimum Degree-Weighted Distance Decoding for Polynomial Residue Codes with Non-Pairwise Coprime Moduli}

\author{\mbox{Li Xiao \,\, and  \,\,  Xiang-Gen Xia}
\thanks{The authors are with Department of
Electrical and Computer Engineering,
University of Delaware, Newark, DE 19716, U.S.A. (e-mail:
\{lixiao, xxia\}@ee.udel.edu).}
}
\maketitle

\begin{abstract}
This paper presents a new decoding for polynomial residue codes, called the minimum degree-weighted distance decoding. The newly proposed  decoding is based on the degree-weighted distance and different from the traditional minimum Hamming distance decoding. It is shown that for the two types of minimum distance decoders, i.e., the minimum degree-weighted distance decoding and the minimum Hamming distance decoding, one is not absolutely stronger than the other, but they can complement each other from different points of view.

\end{abstract}
\begin{IEEEkeywords}
Chinese remainder theorem, error correction, Hamming distance, polynomial residue codes.
\end{IEEEkeywords}

\section{Introduction}\label{sec1}
Polynomial residue codes are a large class of linear codes. Some well-known codes, such as BCH codes, Reed-Solomon codes and Goppa codes \cite{poly4,poly44,goppa}, can be derived from polynomial residue codes. A polynomial residue code with moduli $m_1(x),\cdots,m_L(x)$ encodes a message $a(x)$ as the vector $\mathbf{a}=(a_1(x),\cdots,a_L(x))$, where $\mathbf{a}$ is called a codeword (or residue vector), and $a_i(x)$ are residues of $a(x)$ modulo $m_i(x)$ for $1\leq i\leq L$. The Hamming weight of a codeword in a polynomial residue code is the number of residues that are nonzero, the Hamming distance between two codewords is the number of residues in which they differ, and a polynomial residue code with the minimum Hamming distance $d$ can correct up to $\lfloor(d-1)/2\rfloor$ errors in the residues, where $\lfloor\cdot\rfloor$ denotes the floor function.
In general, polynomial residue codes can be classified into two categories: codes with pairwise coprime moduli \cite{poly1} and codes with non-pairwise coprime moduli \cite{sundaram}.
In a polynomial residue code with pairwise coprime moduli $m_1(x),\cdots,m_N(x),\cdots,m_L(x)$, codewords are
residue vectors of all polynomials with degrees less than that of $\prod_{i=1}^{N}m_i(x)$, where $\mbox{deg}(m_1(x))\leq\cdots\leq\mbox{deg}(m_L(x))$, and the first $N$ moduli form a set of nonredundant moduli, while the last $L-N$ moduli form a set of redundant moduli that facilitates residue error correction. In a polynomial residue code with non-pairwise coprime moduli $m_1(x),\cdots,m_L(x)$, codewords are residue vectors of all polynomials with degrees less than that of the least common multiple (lcm) of all the moduli. Compared with a polynomial residue code with pairwise coprime moduli, a polynomial residue code with non-pairwise coprime moduli has more efficient distributed error decoding and simpler error correction algorithm at the cost of an increase in redundancy. Over the past few decades, polynomial residue codes have been extensively investigated due to their ability of fault-tolerance in polynomial-type computing required for signal processing tasks such as cyclic convolution and FFT computations \cite{CRT2,sundaram,poly1,mura,why,haining,xiaoxia2}. Correspondingly, there is also a great amount of work on integer residue codes \cite{bm,lly2-2002,lly3-2002,vgoh2008,chang,hanshen}.

Recently, different from the Hamming distance, another type of distance called degree-weighted distance for polynomial residue codes is defined, and accordingly, a coding framework based on the degree-weighted distance has been developed for polynomial residue codes with pairwise coprime moduli in \cite{polyyu,polyyu2}.
In this paper, with regard to this degree-weighted distance, we naturally study polynomial residue codes with non-pairwise coprime moduli. We derive the error correction capabilities of these codes and also propose the minimum degree-weighted distance decoding algorithm. Moreover, we give two simple examples to show that for the two types of minimum distance decoders for polynomial residue codes with non-pairwise coprime moduli (i.e., the minimum Hamming distance decoding proposed in \cite{sundaram} and the minimum degree-weighted distance decoding proposed in this paper), one is not absolutely stronger than the other, but they can complement each other from different points of view.

\textit{Notations}: Define by $\mathbb{F}[x]$ the set of all polynomials with coefficients in $\mathbb{F}$, where $x$ is an indeterminate and $\mathbb{F}$ is a field. The highest power of $x$ in a polynomial $f(x)$ is the degree of the polynomial, denoted by $\mbox{deg}\left(f(x)\right)$. All the elements of $\mathbb{F}$ are termed scalars and can be expressed as polynomials of degree $0$. A polynomial is called monic if the coefficient of the highest power of $x$ is $1$, and irreducible if it has only a scalar and itself as its factors. The residue of $f(x)$ modulo $g(x)$ is denoted by $\left| f(x)\right|_{g(x)}$. For a polynomial set $\mathcal{F}=\{f_1(x),\cdots,f_L(x)\}$, we define $\mbox{deg}(\mathcal{F})=\sum_{i=1}^{L}\mbox{deg}(f_i(x))$, and denote the cardinality of $\mathcal{F}$, the greatest common divisor (gcd) and lcm of all the polynomials in $\mathcal{F}$ by $\#(\mathcal{F})$, $\mbox{gcd}(\mathcal{F})$ and
$\mbox{lcm}\left(\mathcal{F}\right)$, respectively. Both $\mbox{gcd}(\cdot)$ and $\mbox{lcm}(\cdot)$ are taken to be monic for the uniqueness. Two polynomials are said to be coprime if their $\mbox{gcd}$ is $1$. Throughout the paper, all polynomials considered are in $\mathbb{F}[x]$, and $\lfloor\cdot\rfloor$ and $\lceil\cdot\rceil$ stand for the floor and ceiling functions.

\section{Preliminaries}\label{sec2}
Let $\mathcal{M}=\{m_1(x),\cdots,m_L(x)\}$ be a set of $L$ non-pairwise coprime moduli, and $M(x)$ be the lcm of all the moduli, i.e., $M(x)=\mbox{lcm}(\mathcal{M})$. We can represent a polynomial $a(x)$ with $\mbox{deg}(a(x))<\mbox{deg}(M(x))$ by its residue vector $\mathbf{a}=(a_1(x),\cdots,a_L(x))$, where $a_i(x)=|a(x)|_{m_i(x)}$ or $a_i(x)\equiv a(x)\mod m_i(x)$, i.e., there exists $k_i(x)$ such that
\begin{equation}\label{folding}
a(x)=k_i(x)m_i(x)+a_i(x)
\end{equation}
with $\mbox{deg}\left(a_i(x)\right)<\mbox{deg}\left(m_i(x)\right)$. Define a set of $L$ pairwise coprime monic polynomials $\{\mu_i(x)\}_{i=1}^{L}$ such that $\prod_{i=1}^{L}\mu_i(x)=M(x)$ and $\mu_i(x)$ divides $m_i(x)$ for each $i, 1\leq i\leq L$.
Then, $a(x)$ can be uniquely reconstructed from its residue vector via the generalized Chinese remainder theorem (CRT) for polynomials \cite{sundaram} as follows:
\begin{equation}\label{crtpoly}
a(x)=\left|\sum_{i=1}^{L}a_i(x)D_i(x)M_i(x)\right|_{M(x)},
\end{equation}
where $M_i(x)=M(x)/\mu_i(x)$, $D_i(x)$ is the modular multiplicative inverse of $M_i(x)$ with respect to $\mu_i(x)$, if $\mu_i(x)\neq 1$, else $D_i(x)=0$. Note that if moduli $m_i(x)$ are pairwise coprime, we have $\mu_i(x)=m_i(x)$, and the above (\ref{crtpoly}) reduces to the standard CRT for polynomials.
Therefore, polynomials with degrees less than $\mbox{deg}\left(M(x)\right)$ and their residue vectors are isomorphic.

\begin{definition}\label{def1}
A polynomial residue code with non-pairwise coprime moduli $\mathcal{M}$ has a message space $\mathcal{S}=\{a(x):a(x)\in\mathbb{F}[x]\mbox{ with }\mbox{deg}(a(x))<\mbox{deg}(M(x))\}$ and consists of residue vectors (called codewords) of all polynomials in $\mathcal{S}$.
\end{definition}

As mentioned before, the redundancy in a polynomial residue code with $L$ pairwise coprime moduli is introduced by the $L-N$ redundant moduli, since a message has degree less than that of the product (or lcm) of the $N$ nonredundant moduli. However, in a polynomial residue code with $L$ non-pairwise coprime moduli, the redundancy is introduced by the common factors of all pairs of moduli, since a message has degree less than that of the lcm of all the moduli, and the degree of the lcm of all the moduli is less than that of the product of all the moduli.

According to $a_i(x)\equiv a(x)\mod m_i(x)$ and $a_j(x)\equiv a(x)\mod m_j(x)$, it is not hard to see that
\begin{equation}\label{check}
a_i(x)\equiv a_j(x)\mod d_{ij}(x),
\end{equation}
where $d_{ij}(x)=\mbox{gcd}(m_i(x),m_j(x))$. We call (\ref{check}) a consistency check between residues $a_i(x)$ and $a_j(x)$. If (\ref{check}) holds, $a_i(x)$ is said to be consistent with $a_j(x)$; otherwise, $a_i(x)$ and $a_j(x)$ appear in a failed consistency check. So, all pairs of residues in a residue vector $\mathbf{a}$ are consistent. If $t$ errors with values $e_{i_1}(x),\cdots,e_{i_t}(x)$ for $\{i_1,\cdots,i_t\}\subset\{1,\cdots,L\}$ have occurred in the transmission, then the received residues will be given by, for $1\leq i\leq L$,
\begin{equation}
\tilde{a}_i(x) =
  \begin{cases}
    a_i(x)+e_i(x), & \quad \mbox{if } i\in\{i_1,\cdots,i_t\}\\
    a_i(x),  & \quad \mbox{otherwise.}\\
  \end{cases}
\end{equation}
The residue errors $e_i(x)$ satisfy $\mbox{deg}(e_i(x))<\mbox{deg}(m_i(x))$. In the following, let us review the minimum Hamming distance, denoted by $d_{minH}$, and a simple residue error correction algorithm for a polynomial residue code with non-pairwise coprime moduli presented in \cite{sundaram}.

\begin{proposition}\cite{sundaram}\label{p1}
For a polynomial residue code in Definition \ref{def1}, write $M(x)$ in the form
\begin{equation}\label{fenie}
M(x)=p_1(x)^{t_1}p_2(x)^{t_2}\cdots p_K(x)^{t_K},
\end{equation}
where $p_i(x)$ are pairwise coprime, monic and irreducible polynomials, and $t_i$ are positive integers. For each $i$, $1\leq i\leq K$, let $d_i$ represent the number of moduli that contain the factor $p_i(x)^{t_i}$. Then, the minimum Hamming distance of the code is $d_{minH}=\min\{d_1,\cdots,d_K\}$.
\end{proposition}

\begin{proposition}\cite{sundaram}\label{p2}
For a polynomial residue code in Definition \ref{def1}, the lcm of any $L-(d_{minH}-1)$ moduli equals $M(x)$. Moreover,
if only $t\leq\lfloor(d_{minH}-1)/2\rfloor$ errors have occurred in the residues, each erroneous residue appears in at least $\lceil(d_{minH}-1)/2\rceil+1$ failed consistency checks, and each correct residue appears in at most $\lfloor(d_{minH}-1)/2\rfloor$ failed consistency checks.
\end{proposition}

One can see from Proposition \ref{p2} that if there are $\lfloor(d_{minH}-1)/2\rfloor$ or fewer residue errors in a polynomial residue code in Definition \ref{def1}, all the error-free residues can be first located through consistency checks for all pairs of residues $\tilde{a}_i(x)$ for $1\leq i\leq L$, and then $a(x)$ can be accurately reconstructed from these obtained error-free residues.
With the above result, we have the following minimum Hamming distance decoding algorithm.
 \begin{enumerate}
   \item For all pairs of residues in the received residue vector, we perform the consistency checks by (\ref{check}), i.e., for $1\leq i,j\leq L, i\neq j$,
   \begin{equation}
   \tilde{a}_i(x)\equiv \tilde{a}_j(x)\mod d_{ij}(x),
   \end{equation}
   where $d_{ij}(x)=\mbox{gcd}(m_i(x),m_j(x))$.
   \item Let
   \begin{multline}
   \hspace{-0.25cm}\mathcal{Z}=\{\tilde{a}_i(x): \tilde{a}_i(x)\mbox{ appears in at most }\lfloor(d_{minH}-1)/2\rfloor\\\mbox{ failed consistency checks for }1\leq i\leq L\}.
   \end{multline}
 If $\#(\mathcal{Z})=0$, i.e., the cardinality of $\mathcal{Z}$ is zero, we claim that the decoding algorithm fails. Otherwise, go to $3)$.
   \item If all elements in $\mathcal{Z}$ are consistent with each other, we use them to reconstruct $a(x)$ as $\hat{a}(x)$ via the CRT for polynomials by (\ref{crtpoly}). Otherwise, $\hat{a}(x)$ cannot be obtained and we claim that the decoding algorithm fails.
 \end{enumerate}

With the above decoding algorithm, it is easy to see that if there are $\lfloor(d_{minH}-1)/2\rfloor$ or fewer residue errors in a polynomial residue code in Definition \ref{def1}, $a(x)$ can be accurately reconstructed, i.e., $\hat{a}(x)=a(x)$. However, if more than $\lfloor(d_{minH}-1)/2\rfloor$ errors have occurred in the residues, the decoding algorithm may fail, i.e., $\hat{a}(x)$ may not be reconstructed, or even though $a(x)$ can be reconstructed as $\hat{a}(x)$, $\hat{a}(x)=a(x)$ may not hold.

Recently, different from the Hamming distance above, a new type of distance called degree-weighted distance for polynomial residue codes is defined, and a coding framework for polynomial residue codes with pairwise coprime moduli has been developed accordingly in \cite{polyyu,polyyu2}. Based on this newly defined degree-weighted distance, we study polynomial residue codes with non-pairwise coprime moduli in the next section of this paper.

\section{Minimum Degree-Weighted Distance Decoding for Polynomial Residue Codes}\label{sec3}
In this section, we first obtain the minimum degree-weighted distance of a polynomial residue code with non-pairwise coprime moduli, and then based on this, the decoding algorithm is also proposed.

In a polynomial residue code in Definition \ref{def1}, for any $a(x)\in\mathcal{S}$, the degree weight of the codeword $\mathbf{a}=(a_1(x),\cdots,a_L(x))$ is defined by
\begin{equation}
w_D(\mathbf{a})=\sum_{i:a_i(x)\neq0}\mbox{deg}(m_i(x)),
\end{equation}
and for any $a(x),b(x)\in\mathcal{S}$, the degree-weighted distance between two codewords $\mathbf{a}$ and $\mathbf{b}$ is defined by
\begin{equation}
w_D(\mathbf{a}-\mathbf{b})=\sum_{i:a_i(x)\neq b_i(x)}\mbox{deg}(m_i(x)).
\end{equation}
Let $d_{minD}$ denote the minimum degree-weighted distance of the code, which is also the smallest degree weight over all nonzero codewords due to the linearity of the code. Then, we have the following result.

\begin{theorem}\label{theo1}
For a polynomial residue code in Definition \ref{def1}, write $M(x)$ in the form (\ref{fenie}). For each $i$, $1\leq i\leq K$, let $\mathcal{M}_i$ be the set of all the moduli that contain the factor $p_i(x)^{t_i}$. Then, the minimum degree-weighted distance of the polynomial residue code is $d_{minD}=\min\{\mbox{deg}(\mathcal{M}_1),\cdots,\mbox{deg}(\mathcal{M}_K)\}$.
\end{theorem}
\begin{proof}
Let $\mathcal{U}$ be any subset of $\mathcal{M}$ satisfying $\mbox{deg}(\mbox{lcm}(\mathcal{U}))<\mbox{deg}(M(x))$. Then, there must exist at least one $\mathcal{M}_i$ for $1\leq i\leq K$ such that $\mathcal{M}_i\bigcap\mathcal{U}=\mathcal{\varnothing}$, where $\mathcal{\varnothing}$ is the empty set. Therefore, we have
\begin{equation}\label{xiya}
\max\limits_{\mathcal{U}\subset\mathcal{M}}\mbox{deg}(\mathcal{U})=\mbox{deg}(\mathcal{M})-\min\{\mbox{deg}(\mathcal{M}_1),\cdots,\mbox{deg}(\mathcal{M}_K)\}.
\end{equation}
For any nonzero $a(x)\in\mathcal{S}$, assume that its residue vector $\mathbf{a}$ has degree weight
$w_D(\mathbf{a})<\min\{\mbox{deg}(\mathcal{M}_1),\cdots,\mbox{deg}(\mathcal{M}_K)\}$.
Define $\mathcal{K}=\{m_i(x):a_i(x)=0 \mbox{ for }1\leq i\leq L\}$. Then, we have
\begin{equation}\label{maomao}
\mbox{deg}(\mathcal{K})>\mbox{deg}(\mathcal{M})-\min\{\mbox{deg}(\mathcal{M}_1),\cdots,\mbox{deg}(\mathcal{M}_K)\}.
\end{equation}
We can write $a(x)$ as $a(x)=c(x)d(x)$, where $c(x)\neq0$ and $d(x)=\mbox{lcm}(\mathcal{K})$. If $\mbox{deg}(d(x))<\mbox{deg}(M(x))$, from (\ref{xiya}) we get $\mbox{deg}(\mathcal{K})\leq\mbox{deg}(\mathcal{M})-\min\{\mbox{deg}(\mathcal{M}_1),\cdots,\mbox{deg}(\mathcal{M}_K)\}$, which is in contradiction with (\ref{maomao}). Therefore, $\mbox{deg}(d(x))=\mbox{deg}(M(x))$. Since $\mathcal{K}\subset\mathcal{M}$, $M(x)$ is divisible by $d(x)$. Moreover, since $\mbox{deg}(d(x))=\mbox{deg}(M(x))$ and $d(x),M(x)$ are monic, we have $d(x)=M(x)$.
Then, from $a(x)=c(x)d(x)$ and $c(x)\neq0$, we have $\mbox{deg}(a(x))\geq\mbox{deg}(M(x))$, which is impossible since $a(x)\in\mathcal{S}$. We thus have $w_D(\mathbf{a})\geq \min\{\mbox{deg}(\mathcal{M}_1),\cdots,\mbox{deg}(\mathcal{M}_K)\}$ for any nonzero $a(x)\in\mathcal{S}$.

We next show that there exists a codeword with exactly the degree weight $\min\{\mbox{deg}(\mathcal{M}_1),\cdots,\mbox{deg}(\mathcal{M}_K)\}$. Without loss of generality, assume that $\min\{\mbox{deg}(\mathcal{M}_1),\cdots,\mbox{deg}(\mathcal{M}_K)\}=\mbox{deg}(\mathcal{M}_1)$. Then, let $a(x)=\mbox{lcm}(\mathcal{M}\backslash\mathcal{M}_1)$, where $\mathcal{M}\backslash\mathcal{M}_1$ denotes the complement of $\mathcal{M}_1$ with respect to $\mathcal{M}$. Since $a(x)$ is not divisible by $p_1(x)^{t_1}$, we can obtain that $\mbox{deg}(a(x))<\mbox{deg}(M(x))$, i.e., $a(x)\in\mathcal{S}$, and that the residues corresponding to the moduli in $\mathcal{M}_1$ are nonzero, while the other residues are equal to zero. So, this codeword has the degree weight $\min\{\mbox{deg}(\mathcal{M}_1),\cdots,\mbox{deg}(\mathcal{M}_K)\}$. Thus, the minimum degree-weighted distance of the code is $\min\{\mbox{deg}(\mathcal{M}_1),\cdots,\mbox{deg}(\mathcal{M}_K)\}$, and we have completed the proof.
\end{proof}

\begin{theorem}\label{theo2}
For a polynomial residue code in Definition \ref{def1}, the lcm of any subset $\mathcal{V}$ of $\mathcal{M}$ equals $M(x)$, if $\mbox{deg}(\mathcal{V})\geq\mbox{deg}(\mathcal{M})-(d_{minD}-1)$. For a residue $\tilde{a}_i(x)$ in the received residue vector, we define the failed consistency check degree of $\tilde{a}_i(x)$ as
\begin{equation}\label{deggg}
\begin{split}
C(\tilde{a}_i(x))&=\mbox{deg}(\{m_j(x): \tilde{a}_j(x) \mbox{ appears in a failed} \\
&\hspace{-0.5cm}\mbox{consistency check with }\tilde{a}_i(x) \mbox{ for } 1\leq j\leq L\}).
\end{split}
\end{equation}
Then, if only $t$ errors $e_{i_1}(x),\cdots,e_{i_t}(x)$ for $\{i_1,\cdots,i_t\}\subset\{1,\cdots,L\}$ satisfying
\begin{equation}\label{cheeee}
\sum_{l=1}^{t}\mbox{deg}(m_{i_l}(x))\leq\lfloor(d_{minD}-1)/2\rfloor
\end{equation}
have occurred in the residues, the failed consistency check degree of each erroneous residue is at least $\lceil(d_{minD}-1)/2\rceil+1$, and the failed consistency check degree of each correct residue is at most $\lfloor(d_{minD}-1)/2\rfloor$.
\end{theorem}
\begin{proof}
According to (\ref{xiya}) in the proof of Theorem \ref{theo1}, it is easy to see that the lcm of any subset $\mathcal{V}$ of $\mathcal{M}$ equals $M(x)$, if $\mbox{deg}(\mathcal{V})\geq\mbox{deg}(\mathcal{M})-(d_{minD}-1)$. By recalling the definitions of $d_1,\cdots,d_K$ in Proposition \ref{p1} and $\mathcal{M}_1,\cdots,\mathcal{M}_K$ in Theorem \ref{theo1}, we know $\#(\mathcal{M}_i)=d_i$ for $1\leq i\leq K$. For each residue error, say $e_{i_1}(x)$, similar to (\ref{fenie}), we write $m_{i_1}(x)$ in the form
\begin{equation}
m_{i_1}(x)=q_1(x)^{k_1}q_2(x)^{k_2}\cdots q_T(x)^{k_T},
\end{equation}
where $\{q_1(x),\cdots,q_T(x)\}\subset\{p_1(x),\cdots,p_K(x)\}$ are pairwise coprime, monic and irreducible polynomials, $k_i$ are positive integers, and $k_i\leq t_j$ if $q_i(x)=p_j(x)$. Since $\mbox{deg}(e_{i_1}(x))<\mbox{deg}(m_{i_1}(x))$, $e_{i_1}(x)$ does not contain at least one of $q_i(x)^{k_i}$ for $1\leq i\leq T$.
Without loss of generality, we assume that $e_{i_1}(x)$ does not contain $q_l(x)^{k_l}$ and $q_l(x)=p_1(x)$. Then, since $m_{i_1}(x)$ and each modulus in $\mathcal{M}_1$ have the common factor $p_1(x)^{k_l}$, the erroneous residue $\tilde{a}_{i_1}(x)$ will be inconsistent with a correct residue over $\mathcal{M}_1$.
If $k_l=t_1$, we know $m_{i_1}(x)\in\mathcal{M}_1$, and we then perform the consistency check between $\tilde{a}_{i_1}(x)$ and each of the other $d_1-1$ residues over $\mathcal{M}_1$.
Since there are $t$ residue errors in total, there are at least $d_1-t$ correct residues over $\mathcal{M}_1$, and thereby
$\tilde{a}_{i_1}(x)$ appears in at least $d_1-t$ failed consistency checks over $\mathcal{M}_1$. So, the failed consistency check degree of $\tilde{a}_{i_1}(x)$ is
\begin{equation}
\begin{split}
C(\tilde{a}_1(x))&\geq\mbox{deg}(\mathcal{M}_1)-\sum_{l=1}^{t}\mbox{deg}(m_{i_l}(x))\\
&\geq d_{minD}-\lfloor(d_{minD}-1)/2\rfloor\\
&=\lceil(d_{minD}-1)/2\rceil+1.
\end{split}
\end{equation}
If $k_l<t_1$, we know $m_{i_1}(x)\notin\mathcal{M}_1$, and we then perform the consistency check between $\tilde{a}_{i_1}(x)$ and each of the $d_1$ residues over $\mathcal{M}_1$. Due to $m_{i_1}(x)\notin\mathcal{M}_1$, the erroneous residue $\tilde{a}_{i_1}(x)$ is not over $\mathcal{M}_1$.
Since there are $t$ residue errors in total, there are at least $d_1-(t-1)$ correct residues over $\mathcal{M}_1$, and thereby
$\tilde{a}_{i_1}(x)$ appears in at least $d_1-t+1$ failed consistency checks over $\mathcal{M}_1$. So, the failed consistency check degree of $\tilde{a}_{i_1}(x)$ is $C(\tilde{a}_1(x))\geq\mbox{deg}(\mathcal{M}_1)-\sum_{l=2}^{t}\mbox{deg}(m_{i_l}(x))>\lceil(d_{minD}-1)/2\rceil+1$.
This analysis holds for each of the erroneous residues, and thus the failed consistency check degree of each erroneous residue is at least $\lceil(d_{minD}-1)/2\rceil+1$. Next, since a correct residue appears in a failed consistency check only if it is being checked with an erroneous residue, the failed consistency check degree of each correct residue is at most $\lfloor(d_{minD}-1)/2\rfloor$ from (\ref{cheeee}). This completes the proof.
\end{proof}

According to the latter part of Theorem \ref{theo2}, if residue errors satisfying (\ref{cheeee}) have occurred, all the error-free residues can be first located through consistency checks for all pairs of residues $\tilde{a}_i(x)$ for $1\leq i\leq L$, and then according to the former part of Theorem \ref{theo2}, these error-free residues contain enough information to reconstruct the correct value of $a(x)$ via the CRT for polynomials. Therefore, we give the minimum degree-weighted distance decoding algorithm as follows.
 \begin{enumerate}
   \item For all pairs of residues in the received residue vector, we perform the consistency checks by (\ref{check}), i.e., for $1\leq i,j\leq L, i\neq j$,
   \begin{equation}
   \tilde{a}_i(x)\equiv \tilde{a}_j(x)\mod d_{ij}(x),
   \end{equation}
   where $d_{ij}(x)=\mbox{gcd}(m_i(x),m_j(x))$.
   \item Let
   \begin{multline}
   \hspace{-0.3cm}\mathcal{Z}=\{\tilde{a}_i(x): \tilde{a}_i(x)\mbox{ has failed consistency check degree}\\ \mbox{by (\ref{deggg}) at most }\lfloor(d_{minD}-1)/2\rfloor\mbox{ for }1\leq i\leq L\}.
   \end{multline}
 If $\#(\mathcal{Z})=0$, i.e., the cardinality of $\mathcal{Z}$ is zero, we claim that the decoding algorithm fails. Otherwise, go to $3)$.
   \item If all elements in $\mathcal{Z}$ are consistent with each other, we use them to reconstruct $a(x)$ as $\hat{a}(x)$ via the CRT for polynomials by (\ref{crtpoly}). Otherwise, $\hat{a}(x)$ cannot be obtained and we claim that the decoding algorithm fails.
 \end{enumerate}

By the above minimum degree-weighted distance decoding algorithm, if residue errors satisfying (\ref{cheeee}) have occurred, $a(x)$ can be accurately reconstructed, i.e., $\hat{a}(x)=a(x)$. We next present two examples to show that none of the minimum Hamming distance decoding proposed in \cite{sundaram} and the minimum degree-weighted distance decoding proposed in this paper is absolutely stronger than the other.

\begin{example}
Let $\mathbb{F}=\mathbb{R}$ be the field of real numbers, and $m_1(x)=(x+1)^2(x+2)^2(x+3)^5, m_2(x)=(x+1)^2(x+2)(x+3)^5(x+4)^2, m_3(x)=(x+2)^2(x+3)^5(x+4)^2, m_4(x)=(x+1)^2(x+2)^2(x+4)^2, m_5(x)=(x+1)(x+2)^2(x+3)(x+4)$. We then have $d_{minH}=3$ and $d_{minD}=25$. Considering the two decoders in Proposition \ref{p2} and Theorem \ref{theo2}, we observe:
\begin{itemize}
  \item The minimum Hamming distance decoding corrects a single error occurring in an arbitrary residue.
  \item The minimum degree-weighted distance decoding also corrects a single error occurring in an arbitrary residue, and in addition, it corrects two errors occurring in the fourth and fifth residues.
\end{itemize}
\end{example}

\begin{example}
Let $\mathbb{F}=\mathbb{R}$ be the field of real numbers, and $m_1(x)=(x+1)^3(x+3)^7(x+4)^2, m_2(x)=(x+1)^3(x+2)(x+3), m_3(x)=(x+2)^2(x+3)^7(x+4)^2, m_4(x)=(x+1)^3(x+2)^2(x+4)^2, m_5(x)=(x+1)(x+2)^2(x+3)^7(x+4)$. We then have $d_{minH}=3$ and $d_{minD}=24$. Considering the two decoders in Proposition \ref{p2} and Theorem \ref{theo2}, we observe:
\begin{itemize}
  \item The minimum Hamming distance decoding corrects a single error occurring in an arbitrary residue.
  \item The minimum degree-weighted distance decoding corrects a single error occurring in anyone of the last $4$ residues, but not in the first residue.
\end{itemize}
\end{example}

\section{Conclusion}
In this paper, we investigated the minimum degree-weighted distance decoding for polynomial residue codes with non-pairwise coprime moduli, which is sometimes but not absolutely stronger than the traditional minimum Hamming distance decoding, and it also provides a new perspective on studying the codes.

\end{document}